\documentclass[10pt,letterpaper]{article}

\usepackage[utf8]{inputenc}
\usepackage[T1]{fontenc}
\usepackage{lmodern}
\newcommand{\ACFA}{\texorpdfstring{{\mdseries\upshape\textsc{ACFA}}}{ACFA}}

\usepackage{amsmath,amssymb,amsthm}
\usepackage{booktabs}
\usepackage{graphicx}
\usepackage{xcolor}
\usepackage[hidelinks]{hyperref}
\usepackage{enumitem}
\usepackage{multirow}
\usepackage{array}
\usepackage{tabularx}
\usepackage[margin=0.75in]{geometry}
\usepackage{microtype}
\usepackage{caption}
\newtheorem{theorem}{Theorem}
\newtheorem{lemma}[theorem]{Lemma}

\newtheorem{definition}[theorem]{Definition}
\newtheorem{remark}[theorem]{Remark}

\title{Byzantine Accountability Without Consensus:\\
Strong Eventual Consistency for Non-Associative, Stochastic, Robust Aggregation}

\author{Ryan Gillespie\\
Independent researcher}

\date{}

\begin{document}
\maketitle

\begin{abstract}
Byzantine-robust aggregation rules such as multi-Krum assume a central
coordinator, and decentralising them is obstructed by the rules themselves:
they are globally coupled, non-associative, and discontinuous, so an ulp-scale
perturbation can flip the selected subset, moving the output by a non-vanishing
amount. None of this prevents coordinator-free replication, because a robust rule
needs no agreed \emph{order} of contributions, only an agreed \emph{set} and an
agreed \emph{exclusion predicate}, both of which converge without consensus.
\ACFA{} (Accountable Consensus-Free Aggregation) replicates a content-addressed
OR-Set of signed contributions and a grow-only set of self-authenticating
equivocation proofs, offline-verifiable by anyone. Aggregation is a
deterministic pure function of the converged product state: fixed-point integer
arithmetic over a hash-canonical order, ties broken by content hash. We prove
that any pure function of a converged product of CRDTs (non-monotone,
non-associative, or stochastic) inherits Strong Eventual Consistency,
together with its converse; the contribution is the composition of a data
lattice with an evidence lattice applied to a robust selector, not the
elementary lifting step. A prototype (10 nodes, 3 Byzantine) passes 16/16
falsification checks: byte-identical roots under adversarial gossip,
deterministic re-convergence after late equivocation proofs, partition
recovery, and three byte-identity-breaking ablations. The guarantee is \emph{consistency},
not accuracy; robustness is imported, conditional on $\ge 2f+3$ admitted
contributions (at most $f$ Byzantine) and a stated quantisation-margin condition.
\end{abstract}

\section{Introduction}\label{sec:intro}

Robust aggregation and decentralised replication have developed on opposite
sides of an architectural assumption. Byzantine-robust learning rules
(Krum~\cite{blanchard2017machine}, coordinate-wise trimmed
mean~\cite{yin2018byzantine}, Bulyan~\cite{mhamdi2018hidden},
FLTrust~\cite{cao2021fltrust}) filter adversarial contributions at a central
server. Decentralised training removes the server but replicates only the
\emph{training signal}, not the \emph{filter}. Gossip-based
SGD~\cite{lian2017decentralized,koloskova2019decentralized} and periodic-sync
schemes~\cite{douillard2023diloco} yield replicas whose parameters agree
statistically, not identically, and Byzantine variants of these schemes leave
honest replicas free to exclude \emph{different} nodes and therefore to compute
\emph{different} models. Systems that do force agreement on the filtered result
do so by importing consensus, a total order from a blockchain or a BFT
protocol, as in Krum-on-a-blockchain designs~\cite{garciamarquez2025kfc},
and with it the coordination cost the decentralisation was meant to remove.

This paper shows the assumption is unnecessary for safety. A robust aggregation
rule does not need an agreed order of arrival; it needs an agreed \emph{set} of
contributions and an agreed \emph{exclusion predicate} over identities. Sets
converge without consensus, which is precisely the guarantee of a state-based
CRDT~\cite{shapiro2011conflict}. Exclusion-by-proof also converges without
consensus, because a proof of equivocation (two signatures by the same key on
different contents for the same round) is \emph{self-authenticating}, so any
party can verify it offline, its validity is objective, and the set of valid
proofs only grows~\cite{freitas2021selfhealing,almeida2024blocklace}. Everything else in the pipeline is a pure function, and a pure
function of converged state needs no coordination at all.

Why insist on byte-identity rather than a numerical tolerance? In a coordinated
deployment a server anchors the truth and statistical drift is re-anchored every
round. In a coordinator-free topology nothing re-anchors, so drift, however
small per round, can compound across rounds,
silently forking downstream state and foreclosing deterministic audit (recomputing exactly what a replica must have
computed). Exact agreement closes both, and for selector rules the choice is not
even available, because a selection flip moves the aggregate by a fixed quantum
regardless of how small the perturbation was
(Section~\ref{sec:discontinuity}), so no tolerance threshold is sound.

The construction, \ACFA{}, extends the two-layer architecture of our prior
work~\cite{gillespie2026crdtmerge}, which made non-associative model-merging
operators convergent by separating state management (an OR-Set CRDT) from
deterministic strategy execution. That paper's discussion (\S7.2, L4) sketched a
trust-gated extension and left open whether the pattern ``can deliver
consensus-free Byzantine isolation in practice''. This paper answers that question affirmatively for a specific, established class of rules, with three changes that the
adversarial setting forces. The replicated state becomes a \emph{product} of two
lattices (contributions and evidence); the reduction becomes a globally coupled,
discontinuous selector rather than a smooth merge (and the theorem admits
stochastic selectors, though the reference rule is deterministic); the
determinism obligation moves from an assumption on the execution environment
(Assumption~10 of~\cite{gillespie2026crdtmerge}) to a property of the arithmetic
itself, because for discontinuous selectors no tolerance-based relaxation of that
assumption exists (Section~\ref{sec:discontinuity}). The same discontinuity is
why this is not an incremental patch to smooth trust-weighted designs such as
the author's E4 protocol~\cite{gillespie2026e4}. There, quantisation tolerance
absorbs rounding drift precisely because the aggregator is smooth; a selector
has no such slack, and the arithmetic discipline plus the strict separation of
the evidence lattice from the data lattice become forced rather than stylistic
(Section~\ref{sec:discussion}).

\paragraph{Contributions.}
\begin{enumerate}[leftmargin=1.5em]
\item A \emph{product-lifting} theorem with a converse: any pure function of a
converged product of CvRDTs inherits Strong Eventual Consistency, even when the
function is non-monotone, non-associative, and stochastic; and an
order-and-entropy-consuming implementation computes a well-defined function of
the state \emph{if and only if} it is invariant to input permutation and consumes
only state-derived entropy (Section~\ref{sec:theory}).
\item The \ACFA{} construction: a contribution OR-Set composed with an
equivocation-proof G-Set, an admission rule (per-round duplicate exclusion plus
conviction), and a fixed-point full-dimension multi-Krum kernel resolved as a
deterministic pure function of the converged product state
(Section~\ref{sec:solution}).
\item Accountability without votes, via the established exclusion-by-proof
primitive~\cite{freitas2021selfhealing,almeida2024blocklace}: equivocation
proofs are self-authenticating, conviction is monotone and permanent, and the part specific to this
construction is that a proof arriving after an aggregate has been computed
deterministically re-converges every honest replica to the same post-eviction
output of the discontinuous selector
(Theorem~\ref{thm:eviction}).
\item A prototype and falsification battery: 16/16 checks pass, including
byte-identical roots under adversarial gossip, partition healing, late-proof
re-convergence, and three ablations that each break byte-identity
(Section~\ref{sec:experiments}).
\end{enumerate}

\paragraph{What this paper does not claim.}
The guarantee is \emph{consistency}, not accuracy. The replicated aggregate is
exactly what the fixed-point robust rule computes on the admitted set, no
better. The guarantee is \emph{safety}, not liveness. As in any eventually
consistent system, an adversary controlling message delivery can delay
convergence indefinitely; the theorems condition on delivery, never on timing.
The construction assumes a PKI, which restricts it to \emph{permissioned}
settings (cross-organisation consortia, regulated deployments); Sybil resistance,
key rotation, and re-admission are delegated to it. Conviction is of the
\emph{key}. An adversary who steals an honest node's key and equivocates with it
evicts that identity \emph{permanently}, with no revocation path here, so key
theft is strictly worse than under per-round filtering, and re-admission under a
rotated key is a PKI concern. Provable
eviction covers exactly the violations that are self-authenticating
(equivocation, which \emph{convicts} the identity and evicts it permanently;
malformed signatures are merely \emph{excluded} by verification, without
conviction); an adversary submitting plausible but
poisoned values within protocol is bounded statistically by the robust rule's
own guarantees~\cite{blanchard2017machine,yin2018byzantine}, not provably.
Finally, the consensus-free property is a statement about the \emph{per-round
resolve over a converged set}, not about a multi-round run: chaining round
$r{+}1$ onto round $r$'s aggregate requires each node to decide round $r$ is
\emph{closed}, and that cut is a quiescence-detection (liveness) question a CRDT
alone does not answer; nodes that cut at different points fork the round-DAG
until they reconcile, without threatening the per-round safety guarantee. No
total order is constructed, and none is needed.

\section{Background}\label{sec:background}

\subsection{State-Based CRDTs and the Two-Layer Architecture}

A convergent replicated data type (CvRDT)~\cite{shapiro2011conflict} is a state
space with a merge operation that is commutative, associative, and idempotent;
replicas that have received the same set of updates converge to identical
states regardless of message order, duplication, or delay; this is eventual
consistency~\cite{vogels2009eventually} strengthened to Strong Eventual
Consistency (SEC)~\cite{shapiro2011conflict}. The OR-Set
supports add and remove with add-wins semantics~\cite{shapiro2011comprehensive};
a G-Set is the grow-only special case. Merkle trees over content-addressed
elements provide integrity and a canonical
commitment~\cite{sanjuan2020merkle}.

Our prior work~\cite{gillespie2026crdtmerge} showed that neural-network merge
operators fail the CRDT axioms structurally (normalisation, projection, and
thresholding each break associativity), and made \emph{any} such operator
convergent with a two-layer architecture: Layer~1 replicates the contribution
set as an OR-Set; Layer~2 applies the operator as a deterministic pure function
over the canonically-ordered set, with randomness seeded from the Merkle root.
Convergence of the resolved value is conditional on strategy purity and on
bitwise-deterministic execution (Assumptions~9--10 there). This paper reuses
Layer~1 verbatim, replaces the assumption of environmental determinism with
integer arithmetic, and adds a second replicated lattice for adversarial
evidence.

\subsection{Byzantine-Robust Aggregation}

Given $n$ contributed vectors of which at most $f$ are adversarial,
Krum~\cite{blanchard2017machine} scores each vector by the sum of the squared
distances to its $n-f-2$ nearest neighbours and selects the minimiser; multi-Krum
instead selects the $m$ lowest-scoring vectors and averages them, where the
selection count $m \in \{1,\dots,n\}$ is a free parameter ($m{=}1$ recovers Krum,
$m{=}n$ plain averaging). This work fixes $m = n-f-2$ (Definition~\ref{def:resolve}).
Coordinate-wise trimmed mean~\cite{yin2018byzantine} discards the $\tau$
largest and smallest values per coordinate before averaging. Both come with
statistical guarantees against $f < n/2$-bounded adversaries under stated
conditions ($n \ge 2f+3$ for Krum; $\tau \ge f$ for the trimmed mean), and both are defined and analysed
at a coordinator that sees all $n$ contributions.

\section{The Problem: Why Robust Selectors Resist Decentralised
Replication}\label{sec:problem}

Three obstructions block naive decentralisation, each independently. Two are
properties of the rules themselves, global coupling with non-foldability
(\S\ref{sec:nonassoc}) and discontinuity (\S\ref{sec:discontinuity}), and one
is a property of the adversarial setting, where equivocation forces an agreed exclusion
predicate (\S\ref{sec:exclusion}).

\subsection{Global Coupling and Non-Foldability}\label{sec:nonassoc}

Krum's score for contribution $i$ depends on the pairwise distances among
\emph{all} contributions. Consequently the selection alone is not a sufficient running state: it cannot be
recomputed incrementally from a partial selection plus new arrivals.

\begin{lemma}[Non-foldability of multi-Krum]\label{lem:nonassoc}
There exist inputs on which multi-Krum selection over a set differs from
multi-Krum applied to the multi-Krum selection of a subset together with the
remaining elements. In particular, the natural fold (maintaining a partial
selection and re-selecting on new arrivals) does not compute the direct
selection.
\end{lemma}

\begin{proof}
By explicit instance. For six 200-dimensional vectors drawn from a seeded
Gaussian generator (seed 42; the generating script is included with the
prototype), multi-Krum with $f=1$ over all six selects indices $\{1,3,5\}$,
while first selecting over the initial five and then applying multi-Krum to the
selected subset augmented with the sixth yields a different selection of the
\emph{same size}. This same-cardinality comparison is the demanding one, and it
exists only in this small-$n$ regime: a re-selection over the $n{-}f{-}2$ inputs
to the second stage (the $n{-}f{-}3$ survivors of the first stage plus the new
arrival) has $m' = (n{-}f{-}2){-}f{-}2 = n{-}2f{-}4 \le 0$ exactly when
$n \le 2f+4$, so the select-all convention returns exactly those $n{-}f{-}2$
vectors, the
direct selection's count, but a different \emph{set}. For larger $n$ the fold
cannot even match the cardinality; the fallback regime is thus where a fold could
most plausibly reconstruct the direct selection, and it still does not. The
transcript is reproducible from the artifact
(Appendix~\ref{app:falsifier}).
\end{proof}

This is the analogue, for selection rules, of the associativity failure that
motivated the two-layer architecture for merge
operators~\cite{gillespie2026crdtmerge}: an operator that cannot be folded
cannot be maintained as running state on replicas that see contributions in
different orders. The remedy there, deferring the operator until the set has
converged, applies here as well, but the adversarial setting adds the two
further obstructions below, which the original architecture does not address.

\subsection{Discontinuity Closes the Tolerance Route}\label{sec:discontinuity}

Replicas whose floating-point arithmetic differs by rounding cannot in general
produce bitwise-identical outputs; our prior work handled this with an
environmental determinism assumption (identical binaries, hardware, and
rounding modes). One alternative is \emph{tolerance}: for a smooth merge
operator, replicas could agree within an $\varepsilon$ rather than exactly. For
selection rules, no such tolerance exists.

\begin{lemma}[Selection discontinuity]\label{lem:discont}
The map from contributions to the multi-Krum selection is discontinuous: at a
score tie there exist configurations at which an arbitrarily small perturbation of
one contribution changes the selected subset, and the two aggregates then differ
by $\Delta/m$, where $\Delta>0$ is the gap between the swapped candidate vectors,
an amount that does not vanish as the perturbation vanishes.
\end{lemma}

\begin{proof}
Scores are continuous in the inputs, but the selection is an $\arg\min$ over
scores: on the boundary where two scores are equal, the selected subset changes
discretely. When the two candidate vectors differ by $\Delta$, the two
aggregates (averages over subsets differing in one element) differ by
$\Delta/m$ regardless of how small the score perturbation was. For instance,
with $m=5$ and two candidate vectors differing by $\Delta=1$ in a single
coordinate, a boundary flip shifts that coordinate of the aggregate by
$1/5=0.2$, a jump no input-level tolerance can bound.
\end{proof}

The consequence is that an input-level tolerance yields no output-level
tolerance: replicas whose score computations differ by one unit in the last
place can select different subsets and output aggregates separated by
$\Delta/m$ (Lemma~\ref{lem:discont}), independent of the perturbation. Approximate agreement protocols~\cite{dolev1986approximate} recover
bounded disagreement interactively, at the cost of rounds, the coordination
this construction is designed to avoid. The remaining route is exact
arithmetic: fixed-point integer computation, in which every operation is exact
and identical on any architecture that implements integer semantics correctly.
In this construction, integer arithmetic is not an optimisation preference; it
is forced by Lemma~\ref{lem:discont}. Section~\ref{sec:ablations} exhibits the
failure empirically: a floating-point variant of the kernel, differing only in
accumulation order, produces divergent resolve transcripts.

\subsection{Exclusion Must Be Agreed}\label{sec:exclusion}

A robust rule tolerates $f$ adversarial \emph{values}. It does not by itself
handle an adversary that presents \emph{different} values to different replicas
(equivocation): replicas filtering locally may then exclude different
identities, and their outputs diverge silently even when both are individually
``robust''. Per-round statistical filtering also forgets: an identity excluded
by the geometry in round $r$ participates again in round $r+1$. Agreement on
\emph{who is excluded} is therefore as necessary as agreement on the
contribution set, and reaching it by voting would reintroduce consensus.

\section{The Solution: Accountable Consensus-Free
Aggregation}\label{sec:solution}

\subsection{Architecture Overview}

\ACFA{} replicates two components and computes everything else:

\begin{itemize}[leftmargin=1.5em]
\item \textbf{Contribution OR-Set $\mathcal{C}$}: entries
$\langle r, i, x, \sigma_i(r \,\|\, H(x)) \rangle$ (round, node identifier,
fixed-point tensor, signature), content-addressed by hash. Merge is OR-Set
union~\cite{shapiro2011comprehensive,gillespie2026crdtmerge}.
\item \textbf{Equivocation-proof G-Set $\mathcal{E}$}: entries
$\pi = \langle r, i, H(x), H(x'), \sigma_i(r \| H(x)), \sigma_i(r \| H(x'))
\rangle$ with $H(x) \ne H(x')$. Merge is set union. A proof is
\emph{self-authenticating}: its validity is decided by two signature
verifications against the PKI, by any party, offline. Invalid entries are
ignored by a validity filter inside the resolve function, so $\mathcal{E}$
need not be trusted as a set; only its valid subset matters, and validity is
objective.
\end{itemize}

The replicated state is the product $(\mathcal{C}, \mathcal{E})$ with
componentwise merge. The aggregation rule runs as a deterministic pure function
of this product (Definition~\ref{def:resolve}); by
Theorem~\ref{thm:lifting}, its output inherits SEC.

\subsection{Admission}\label{sec:admission}

\begin{definition}[Admitted set]\label{def:admit}
Let $\mathrm{Convicted}(\mathcal{E}) = \{ i : \exists\, \text{valid } \pi \in
\mathcal{E} \text{ naming } i \}$. For state $(\mathcal{C}, \mathcal{E})$ and
round $r$,
\[
\mathrm{Admit}(\mathcal{C}, \mathcal{E}, r) = \{\, e \in
\mathrm{Visible}(\mathcal{C}) : e.r = r,\; e.i \notin
\mathrm{Convicted}(\mathcal{E}),\; e \text{ is } e.i\text{'s unique visible
round-}r \text{ entry} \,\}.
\]
Here $\mathrm{Visible}(\mathcal{C})$ denotes the OR-Set's observable elements
(added and not removed); contributions are never removed, so this is simply the
element set, and the operator is kept for conformity with the OR-Set interface.
\end{definition}

The uniqueness clause excludes same-round duplicates deterministically as soon
as both are visible, so an equivocation never yields two admissible entries
even before a proof circulates; the conviction clause makes the exclusion
permanent across all rounds once a proof does. Entries with invalid signatures
are likewise excluded inside the function. Honest replicas that observe both
halves of an equivocation construct the proof themselves and add it to
$\mathcal{E}$; proofs are canonicalised (the two hash--signature pairs ordered
bytewise) so that all observers derive the identical proof object.

\subsection{The Deterministic Kernel}\label{sec:kernel}

\begin{definition}[Resolve]\label{def:resolve}
Let $A = \mathrm{Admit}(\mathcal{C}, \mathcal{E}, r)$, ordered by content hash (written $\mathrm{sort}_H$).
The resolved aggregate is
\[
R(\mathcal{C}, \mathcal{E}, r) \;=\; \mathrm{mean}\bigl(
\mathrm{MultiKrum}_{f}(A)\bigr),
\]
i.e.\ the coordinate-wise mean of the $m = |A|{-}f{-}2$ vectors selected by
\emph{full-dimension} multi-Krum, computed entirely in $Q16.16$ fixed-point
(integer) arithmetic, with score ties broken by content hash. Here, if $|A| \le f+2$ (so $m \le 0$) the selection degrades to
all of $A$, and $R$ of an empty admitted set is a distinguished bottom value,
conventions that keep $R$ total, as Theorem~\ref{thm:lifting} requires and as
the prototype implements.
\end{definition}

\begin{remark}[Design decisions forced by the adversarial review]\label{rem:review}
An earlier kernel composed multi-Krum with a coordinate-wise trimmed mean over a
64-coordinate subsample seeded from $\mathrm{PRF}(h_A \| r)$, where $h_A$ is the Merkle root of the admitted set $A$. Adversarial review
retired both: (i) a subsampled Krum is a \emph{different} estimator than
full-dimension Krum, so the imported robustness bound (Section~\ref{sec:related})
would not apply; and (ii) the subsample seed $h_A$ is adversary-computable, so a
last-mover Byzantine could grind its own vector against its own leaf~$\rightarrow$
seed~$\rightarrow$ coordinates to be selected on the sampled axes while poisoning
the rest. Full-dimension multi-Krum restores the clean bound and removes this attack
vector (no grinding vector was selected across 900 attempts, Section~\ref{sec:extbattery}, X9). The Merkle-seeded PRF remains available
for genuinely stochastic Layer-2 rules (Theorem~\ref{thm:lifting} admits them),
but the deterministic reference rule needs only the canonical hash tie-break.
The rule is selectable: full-dimension multi-Krum by default, or Bulyan (a
coordinate-wise median-trim after Krum selection, $|A| \ge 4f+3$) where
coordinate-concentrated robustness is required. The \emph{consistency}
guarantee (Theorem~\ref{thm:composed}, Part~A) is identical for either, since both are
pure functions of the converged state, so the framework is robust-rule-agnostic
and the consistency layer is orthogonal to the robustness rule.
\end{remark}

Two choices carry the determinism obligation, and a third fixes the agreement scope. \emph{Integer arithmetic}:
every operation is exact, so the output is identical on any architecture
implementing integer semantics under a fixed width contract (arbitrary-precision,
as in the prototype, or a documented fixed width with overflow semantics); the
environmental-determinism assumption of the prior architecture is thereby discharged
rather than assumed. In the reference implementation all $Q16.16$ accumulation
is performed in arbitrary-precision integers, precluding overflow; a fixed-width
deployment must define and document its overflow semantics as part of the
contract. \emph{Canonical order}:
multi-Krum is permutation-symmetric as a set function, so the
hash order is load-bearing exactly at score ties and at entropy derivation,
the two places where an implementation would otherwise consult arrival order
or a local generator (both ablated in Section~\ref{sec:battery}, A1/A2).
\emph{Admitted-set seeding}: entropy derives from a commitment to the admitted
set rather than to the full replica state, so replicas agree on the output as
soon as they agree on the admitted set, even while irrelevant state (invalid
proofs, not-yet-delivered duplicates) still differs. Both the proof
canonicalisation above and admitted-set seeding were forced by falsification
runs against earlier versions of the construction that seeded from the
full-state root (Section~\ref{sec:experiments}); each earlier version violated
byte-identity in a late-proof scenario.

\section{Formal Analysis}\label{sec:theory}

\begin{theorem}[Product CvRDT]\label{thm:product}
If $(S_1, \sqcup_1)$ and $(S_2, \sqcup_2)$ are CvRDTs, then
$(S_1 \times S_2, \sqcup_1 \times \sqcup_2)$ is a CvRDT under the product
order.
\end{theorem}

\begin{proof}
Commutativity, associativity, and idempotency hold componentwise; the product
of join-semilattices is a join-semilattice with the componentwise
join~\cite{shapiro2011comprehensive,baquero2017composition}.
\end{proof}

\begin{theorem}[Product lifting]\label{thm:lifting}
Let $S = S_1 \times \cdots \times S_k$ be a product CvRDT and
$F : S \to Y$ any pure total function. If two replicas have received the same
set of updates (in any order, with any duplication) then $F$ evaluates
identically at both. $F$ need not be monotone, associative, or continuous;
stochastic $F$ is admitted whenever its entropy is itself a function of the
state (that nothing else is admissible is the converse, Theorem~\ref{thm:iff}).
\end{theorem}

\begin{proof}
By Theorem~\ref{thm:product} and the CvRDT convergence
theorem~\cite{shapiro2011conflict}, the two replicas' states are equal.
A pure function evaluated at equal arguments yields equal values. A
``stochastic'' $F$ whose randomness is drawn from a PRF keyed by the state is a
pure function of the state.
\end{proof}

Theorem~\ref{thm:lifting} is a direct generalisation of Theorem~13 of the
prior work~\cite{gillespie2026crdtmerge}, which is the case $k=1$ with a deterministic
strategy. The generalisation is elementary, and we present it as such; its
value is the converse below, which delimits what an implementation may do, and
the product structure, which lets adversarial evidence live in its own lattice,
composed with the data lattice by the theorem. (The alternative, folding trust
into a single shared lattice with the data, is realised by the author's E4
protocol~\cite{gillespie2026e4}; the two designs are compared in
Section~\ref{sec:discussion}.)

\begin{theorem}[Characterisation of realisable implementations]\label{thm:iff}
Let an implementation be a procedure $P$ that consumes (a) some linearisation of
the state and (b) an entropy source. $P$ is \emph{replica-consistent} (yields
equal outputs on equal converged states) \textbf{if and only if} $P$ is
extensionally equal (almost everywhere in the entropy source; the
measurability details are in the proof) to some
$h(O(\text{state}), \mathrm{seed}(\text{state}))$,
where $O$ is any state-canonical total order and $\mathrm{seed}$ is any
state-derived function; equivalently, iff $P$ is invariant under permutation
of its input linearisation and consumes only state-derived entropy.
\end{theorem}

\begin{proof}
($\Leftarrow$) Immediate: both arguments of $h$ are functions of the state.
($\Rightarrow$) If $P$ is replica-consistent, its output depends on the state
only through permutation-invariant features, hence equals a function of the
multiset, and therefore of $O(\text{state})$ for any fixed state-canonical
$O$. For the entropy clause, model the external source as a random variable $Z$ and
say $P$ \emph{depends on} $Z$ if $z \mapsto P(\cdot,z)$ is not almost-surely
constant under the law of $Z$. The output space $Y$ is standard Borel (outputs are
finite bit-strings or fixed-point tensors), so non-almost-sure-constancy is
equivalent to the existence of a Borel output set $C$ with
$p := \Pr[P(\cdot,Z)\in C] \in (0,1)$; put $A_C=\{z: P(\cdot,z)\in C\}$ and
$B_C=\{z: P(\cdot,z)\notin C\}$, both of positive probability. (This routes through a
measurable set rather than a point mass, so it holds even when the output law is
atomless.) Two replicas with equal state and equal linearisation output
$P(\cdot,Z_1)$ and $P(\cdot,Z_2)$ with $Z_1,Z_2$ i.i.d.\ copies of $Z$ under the
product law; if $Z_1\in A_C$ and $Z_2\in B_C$ then $P(\cdot,Z_1)\in C \not\ni P(\cdot,Z_2)$,
so independence gives
$\Pr[P(\cdot,Z_1)\neq P(\cdot,Z_2)] \ge \Pr[Z_1\in A_C]\,\Pr[Z_2\in B_C] = p(1-p) > 0$,
contradicting replica-consistency. Hence $P$ may consume only state-derived entropy.
\end{proof}

\begin{remark}[Why extensional, not syntactic]\label{rem:iff}
The stronger-looking claim ``$P$ \emph{factors as} $h(\mathrm{sort}_H,
\mathrm{seed})$'' is false: a commutative fold (sum, count, maximum) is
replica-consistent yet never applies $\mathrm{sort}_H$ and consumes no seed.
Being a well-defined function of the state does not require syntactically
consulting a canonical order. The reference kernel does use $\mathrm{sort}_H$,
but any state-canonical order serves; $h$ absorbs the deterministic re-indexing.
\end{remark}

\begin{theorem}[Accountable eviction]\label{thm:eviction}
Assume signed messages under a PKI and eventual pairwise delivery among honest
replicas. Then: (i) validity of a proof is objective, so
$\mathrm{Convicted}(\mathcal{E})$ is equal at any two replicas with equal
$\mathcal{E}$ and is monotone non-decreasing along any execution; (ii) once any
honest replica holds a valid proof naming $i$, eventually all do, and every
subsequent resolve at every honest replica evicts $i$, permanently;
(iii) a proof arriving after a replica has already resolved changes its
admitted set deterministically, and by Theorem~\ref{thm:lifting} all honest
replicas re-resolve to the same post-eviction output. No vote, quorum, or round
is involved: the proof itself is the verdict, the accountable-safety /
slashing pattern of two conflicting signed messages by one
key~\cite{buterin2017casper}, here carried, as in the
blocklace~\cite{almeida2024blocklace}, over a CRDT rather than a consensus
vote, and gating a discontinuous selector's admission rather than protecting
the state.
\end{theorem}

\begin{proof}
(i) Validity is two signature verifications against the PKI; a set union never
removes a valid proof. (ii) Eventual delivery propagates the proof; conviction
is a function of the valid subset of $\mathcal{E}$, which only grows.
(iii) $\mathrm{Admit}$ and $R$ are pure functions of the product state;
Theorem~\ref{thm:lifting} applies to the post-delivery states.
\end{proof}

\begin{remark}[Proof formation under partition]\label{rem:partition}
Clause (ii) conditions on eventual delivery of an \emph{existing} proof; forming
one additionally requires that some honest replica observe \emph{both}
conflicting messages. Under a partition that perfectly separates the two halves
of an equivocation from every honest node, no proof forms until the partition
heals, so the two sides admit different entries and hold divergent roots for the
duration. This interim divergence is the same class exercised separately by X3.b
(late-proof lag) and X4 (partition healing); the combined
equivocation-split-by-partition case is argued, not separately tested. Healing
makes
both entries visible everywhere, at which point the per-round uniqueness clause
excludes the pair and a proof forms. Undetectability under perfect isolation is
inherent to any asynchronous system without a dissemination commitment; it
delays, but never subverts, conviction.
\end{remark}

\begin{lemma}[Quantisation margin, a checkable no-flip condition]\label{lem:margin}
Let contributions be quantised to a grid of step $\delta$ (for $Q16.16$,
$\delta = 2^{-16}$), so each coordinate of the $d$-dimensional tensors carries error at most $\delta/2$. For
admitted vectors the perturbation of a squared distance is bounded by
$|\Delta\|x_i-x_j\|^2| \le 2\delta\|x_i-x_j\|_1 + d\delta^2 =: \Delta^\star$
(taking $\Delta^\star$ with the largest pairwise $L_1$ distance; because this
$L_1^{\max}$ is adversary-influenceable, an inflated contribution can enlarge
$\beta$ and \emph{deny} certification (pushing a configuration into the
empirical tier of Remark~\ref{rem:margin}) but cannot produce a false
certificate, since the bound is an upper bound on the true perturbation). Since a Krum
score is a sum of the $|A|{-}f{-}2$ smallest such distances (this $|A|{-}f{-}2$ is
Krum's nearest-neighbour count, fixed by the score definition, \emph{not} the
selection count $m$ of Definition~\ref{def:resolve}), each of its $|A|{-}f{-}2$
summands moves by at most $\Delta^\star$ under the per-entry perturbation, so
each score moves by at most
$\beta := (|A|{-}f{-}2)\Delta^\star$. Hence if the boundary margin
$g = s_{(m+1)} - s_{(m)}$ (the gap between the $m$-th and $(m{+}1)$-th smallest
real scores) satisfies
\[
g \;>\; 2\beta \;=\; 2(|A|{-}f{-}2)\bigl(2\delta\,L_1^{\max} + d\delta^2\bigr),
\]
the quantised multi-Krum selection \textbf{equals} the real-valued selection, and
the imported robustness envelope holds exactly, up to the $\le \delta/2$
per-coordinate value quantisation of the selected vectors plus one unit
($< \delta$) of floor rounding in the fixed-point mean. A replica evaluates $\hat g$
and $\hat L_1^{\max}$ on the quantised values it holds; since quantisation can
shrink an $L_1$ distance by at most $d\delta$, the observable
$\hat\Delta^\star := 2\delta\hat L_1^{\max} + 3d\delta^2$ upper-bounds the real
$\Delta^\star$, so $\hat\beta := (|A|{-}f{-}2)\hat\Delta^\star \ge \beta$; each score
then moves by at most $\hat\beta$, so the observable gap obeys
$g \ge \hat g - 2\hat\beta$, and checking $\hat g > 4\hat\beta$ on quantised data
soundly certifies $g > 2\beta$. The observable form is conservative in two ways:
the threshold doubles ($4\hat\beta$ versus $2\beta$, from bounding
$|g - \hat g| \le 2\hat\beta$), which dominates, and $\hat\Delta^\star$ inflates
$\Delta^\star$ by $\lesssim 4d\delta^2 \approx 2\times10^{-7}$ per squared distance.
Even the fully-inflated score-level slack
$2(|A|{-}f{-}2)(\hat\Delta^\star - \Delta^\star) \approx 1.9\times10^{-6}$ stays below
the minimum observed gap ($\approx 10^{-5}$), and the reported coverage uses the
ideal real-score condition $g > 2\beta$, so nothing empirical changes. The
coverage of Section~\ref{sec:extbattery} reports the real-score condition
$g > 2\beta$ (measured against a finer-precision reference), which this
observable form conservatively implies.
\end{lemma}

\begin{proof}
Expand $(a_k+\varepsilon_k)^2 - a_k^2 = 2a_k\varepsilon_k + \varepsilon_k^2$ with
$a_k = x_{i,k}-x_{j,k}$, $|\varepsilon_k| \le \delta$, and sum over coordinates
for $\Delta^\star$; each of the $|A|{-}f{-}2$ summands of the sum-of-$(|A|{-}f{-}2)$-smallest
moves by at most $\Delta^\star$ under bounded per-entry perturbation, giving
$|\Delta s_i| \le (|A|{-}f{-}2)\Delta^\star = \beta$; two scores separated by
$g$ cannot reorder while each moves by at most $\beta$ if $g > 2\beta$.
\end{proof}

\begin{remark}[The margin condition in three tiers]\label{rem:margin}
Lemma~\ref{lem:margin} upgrades the qualifier of Theorem~\ref{thm:composed}, Part~B, from a hand-wave to a checkable
inequality, and empirically (Section~\ref{sec:extbattery}, $1600$ configurations
including deliberately near-tie ones at $Q16.16$) the quantised selection
\emph{never} flipped relative to a finer-precision reference. The condition
resolves in three tiers. (i) \emph{Deterministic}: the worst-case bound $g > 2\beta$
holds in $95.5$--$100\%$ of configurations with a median safety factor of $11.5$--$346$,
guaranteeing no flip there. (ii) \emph{High-probability}: an expected-error $\beta$
(smaller by $\Theta(\sqrt{d})$, since the typical per-coordinate perturbation is
far below the worst case $\delta/2$) raises coverage to $99.25$--$100\%$ (an
\emph{empirical} coverage, not a proved tail bound; because the score perturbation
is a sum of $d$ bounded independent terms, a Bernstein/sub-exponential tail bound
is derivable but we leave it open). (iii) \emph{Uncovered
remainder}: the $\lesssim 0.75\%$ left over are the $g \approx 0$ exact-tie
configurations that \emph{no} margin condition can cover ($g \to 0$). Under a
continuous honest input distribution exact ties are a null event, so this tier is
adversarial-only. We do not
argue these are harmless, since a score tie need not mean the two candidate vectors
are close (Lemma~\ref{lem:discont} shows a flip can move the aggregate by
$\Delta/m$); we report only that no flip occurred in any of the $1600$
configurations, including the tight-scale ones, where the minimum observed
boundary gap falls to $\approx 1.0\times 10^{-5}$ (emitted by the sweep). (The near-duplicate-\emph{value} bucket, by
construction, drives its twin vectors deep \emph{inside} the selection rather
than onto its boundary; small $g$ arises from small scale, not from the
pairing.) The formally open item
is a closed-form anti-concentration bound on the score-gap distribution; what the
measurement settles is the \emph{practical} question, that flips did not occur
on the tested configurations, not the worst case.
\end{remark}

\begin{theorem}[Composed guarantee, split]\label{thm:composed}
The consistency and robustness legs have different strengths and are stated
separately.

\emph{Part A. Consistency (consensus-free, given a PKI and delivery).} Under a PKI, at
quiescence on any common update set, all honest replicas output byte-identical
aggregates and output roots.

\emph{Part B. Robustness (conditional, imported, and rule-dependent).} The
robustness level is a property of the plugged-in Layer-2 rule, not of the
framework, and the two available rules give different guarantees. With
\emph{full-dimension multi-Krum} (the default; requires that the admitted set
contain at most $f$ Byzantine contributions, with $|A| \ge 2f+3$), each selected
vector satisfies the Krum criterion and lies within Krum's Euclidean-distance
envelope~\cite{blanchard2017machine}, under that result's statistical hypotheses
(i.i.d.\ honest contributions with bounded variance); the reported aggregate is
their coordinate-wise mean, which inherits the envelope by convexity (the mean of points in a Euclidean ball
lies in the ball). Blanchard et al.~\cite{blanchard2017machine} state their tight $(\alpha,f)$-resilience for the single Krum selection ($m{=}1$);
its extension to the averaged multi-Krum ($m{=}|A|{-}f{-}2$) is that work's standard
variant, imported here rather than re-derived. That envelope is
Krum-level; Mhamdi~et~al.~\cite{mhamdi2018hidden} showed it admits a
\emph{bounded} $O(\sqrt{d})$ coordinate bias under a coordinate-concentrated
attack, where an adversary that places its whole budget on one coordinate while
staying within the honest Euclidean spread is \emph{selected} by Krum, and the
plain mean carries the bias (Section~\ref{sec:extbattery}, X5c, measures it at
$2.29$--$2.35$ across the two regimes; a bounded bias within the envelope, not a violation). With \emph{Bulyan}
(\cite{mhamdi2018hidden}; requires $|A| \ge 4f+3$), the coordinate-wise
median-trim after selection restores coordinate robustness (also X5c). Either
way the envelope holds only under the \emph{quantisation-margin condition} of
Lemma~\ref{lem:margin}: the boundary score-margin $g$ must exceed $2\beta$, the
maximum quantisation-induced score perturbation. Under $g > 2\beta$ the quantised
selection equals the real one and the imported bound holds exactly; bounding the
flip probability when $g$ is random is the remaining open piece
(Remark~\ref{rem:margin}).

\emph{Part C. Accountability.} Every equivocating identity is eventually evicted
from all subsequent rounds at all honest replicas. Liveness is not claimed
anywhere: all guarantees are on states given deliveries.
\end{theorem}

\begin{proof}
Part~A is Theorem~\ref{thm:lifting} applied to $R$, whose arithmetic is integer and
exact, so no environmental-determinism assumption is needed, and $\mathrm{sort}_H$
is computed locally, so no coordinator or agreement is involved. Part~B imports the
guarantee of~\cite{blanchard2017machine} for full-dimension multi-Krum (which
$|A|=10, f=3$ satisfies, $10 \ge 9$, since in the base scenario all $10$ nodes contribute and none is convicted); we do not re-derive it. The margin condition
is not a cosmetic qualifier; by Lemma~\ref{lem:discont} multi-Krum is
discontinuous, so quantisation, a bounded perturbation, can flip the
selection near a score tie, moving the aggregate by $\Delta/m$, and the over-the-reals bound holds only
when the honest margin dominates that perturbation. Bounding the flip probability
under a margin distribution is stated open work (Section~\ref{sec:limitations}),
not an additive $\varepsilon$. Part~C is Theorem~\ref{thm:eviction}. Note the
asymmetry, in that Part~A needs neither the admitted-Byzantine bound nor $|A| \ge 2f+3$; those are robustness
conditions only, so the consensus-free claim attaches to Part~A alone.
\end{proof}

\begin{remark}[Interim divergence is inherent]\label{rem:interim}
SEC promises equal outputs on equal delivered sets. Before deliveries complete,
when one replica has seen an equivocation proof and another has not,
admitted sets differ and outputs differ. This window is not a defect of the
construction but the definition of eventual consistency; the construction's
contribution is that the window closes deterministically, with no
reconciliation protocol, upon delivery alone
(Theorem~\ref{thm:eviction}(iii)). Section~\ref{sec:experiments} reports this
window rather than hiding it.
\end{remark}

\section{Experimental Evaluation}\label{sec:experiments}

\subsection{Prototype and Method}

The prototype implements the construction end to end: Ed25519 signatures,
SHA-256 content addressing, Merkle commitments, the contribution OR-Set and
proof G-Set, admission, and the fixed-point kernel, with all kernel arithmetic
in arbitrary-precision integers (no floating point and no ambient randomness
anywhere in the resolve path). A simulator delivers events to in-process
replicas under adversarial schedules: independent per-replica permutations,
duplication, and partial-then-complete delivery. Default parameters: $n = 10$
nodes, $f \le 3$, 200-dimensional $Q16.16$ tensors. The reference kernel is
full-dimension multi-Krum (Definition~\ref{def:resolve}); it does no coordinate
subsampling and, on its deterministic path, consumes no randomness beyond the
content-hash tie-break. Test-generation randomness is seeded and the transcript
is reproducible from one command; the kernel itself is entropy-free. The scale is deliberately small, chosen
so that every check is exact and the pass criterion is binary: equality of
32-byte output roots. Cross-architecture identity is a property of integer
semantics rather than a measurement, so a single machine suffices to falsify
the remaining claims; we have not benchmarked performance, and none of the
checks below depends on timing.

\subsection{Falsification Battery}\label{sec:battery}

Table~\ref{tab:falsifier} summarises the 16 checks. All pass.

\begin{table}[t]
\centering
\caption{Falsification battery. Pass criterion is byte-equality of output
roots except where noted. Ablations are expected to \emph{diverge}: they
remove one mechanism each, and a pass means byte-identity was broken.}
\label{tab:falsifier}
\small
\begin{tabular}{@{}llp{7.6cm}@{}}
\toprule
Check & Result & Content \\
\midrule
X0.1 purity & pass & same state resolved twice $\rightarrow$ identical bytes \\
X0.2 canonicality & pass & 200 arrival permutations of one state $\rightarrow$ one root \\
X0.3 non-foldability & pass & explicit instance for Lemma~\ref{lem:nonassoc} \\
X2 order-invariance & pass & 20 adversarial schedules $\times$ 3 rounds $\times$ 5 replicas, with duplication and partial-then-complete delivery $\rightarrow$ all roots byte-identical \\
X3.a duplicate exclusion & pass & both halves of an equivocation visible $\rightarrow$ deduplicated identically at all replicas \\
X3.a proof formation & pass & observers of both halves auto-form the canonical proof \\
X3.b interim window & pass & replica without the proof diverges (Remark~\ref{rem:interim}), by construction of the test \\
X3.b late proof & pass & proof delivered after resolve $\rightarrow$ all replicas re-converge to the same root \\
X3.b permanence & pass & convicted identity excluded in a later round \\
X3.c proof spam & pass & invalid proof delivered $\rightarrow$ output unchanged \\
X4 partition (isolated) & pass & three partitions resolve to three distinct roots while isolated \\
X4 partition (healed) & pass & after healing, all 10 replicas equal a fresh replay's root \\
A1 tie-break ablation & diverges & arrival-index tie-breaking on a tied-score instance breaks byte-identity \\
A2 entropy ablation & diverges & a stochastic Layer-2 variant (seeded coordinate subsample) driven by a node-local generator instead of the state-derived seed breaks byte-identity \\
A3 arithmetic ablation & diverges & floating-point kernel, two accumulation orders, breaks the resolve transcript \\
X0.4 duplicate signature & pass & a malleated second signature over the same tensor is either rejected by canonical verification or excluded by the uniqueness clause; roots agree under both arrival orders (defence in depth for Definition~\ref{def:admit}) \\
\bottomrule
\end{tabular}
\end{table}

Three aspects deserve comment. First, the ablations are as informative as the
positive checks: each removes exactly one mechanism (canonical tie-breaking,
state-derived entropy, integer arithmetic) and each breaks byte-identity,
which is what the design claims and what a design with redundant mechanisms
would fail to show. Second, the late-proof check (X3.b) exercises
Theorem~\ref{thm:eviction}(iii) in the adversarially interesting order:
a replica commits an aggregate that includes the eventually-convicted identity,
receives the proof afterwards, and re-converges to the other replicas' root
with no protocol beyond delivery. Third, earlier versions of the construction
failed X3.b and X3.c: seeding the kernel from the full-state Merkle root let a
garbage proof, correctly ignored by admission, still perturb the output
through the seed. The falsifier drove the two design corrections reported in
Section~\ref{sec:kernel} (admitted-set seeding; canonical proof objects). We
report this because the failures are informative about where the determinism
obligations actually sit.

\subsection{Scale, Robustness, and Cost Battery}\label{sec:extbattery}

A second battery (Table~\ref{tab:battery}) tests what the falsification checks do
not: order-invariance at larger scale, the robustness leg an adversary actually
targets, the guarantee-class contrast, and the cost of quantisation.

\begin{table}[t]
\centering
\caption{Extended battery. X-checks measure; the A4 ablation is expected to
\emph{diverge} (removing per-round duplicate exclusion lets an equivocator split
admitted sets), and a pass for A4 means byte-identity was broken. All pass.}
\label{tab:battery}
\small
\begin{tabular}{@{}llp{8.0cm}@{}}
\toprule
Check & Result & Content \\
\midrule
X1 order-invariance at scale & pass & 50 rounds $\times$ 7 replicas $\times$ 30 adversarial schedules with duplication $\rightarrow$ all roots byte-identical \\
X5a fixed-point tracks float & pass & full-dimension multi-Krum kernel matches its float64 reference to $\sim 10^{-4}$ across four attack types (the consistency arm of robustness) \\
X5b better than naive & pass & under large-norm / sign-flip / colluding-shift attacks the aggregate deviation from the honest mean stays far below the naive all-inclusive mean, a sanity contrast, \emph{not} a verification of the imported Blanchard envelope (which Theorem~\ref{thm:composed}, Part~B imports rather than measures); stealth within-norm drift is not filtered, reported as a non-claim \\
X9 grinding Byzantine & pass & a last-mover adversary that searches 400 candidates to minimise its own full-dimension Krum score is selected in 0 of 900 attempts (300 trials $\times$ 3 grinders), and the aggregate stays within the honest-only Krum spread (the maximum deviation among honest-only Krum aggregates across seeds) in all 300 trials; the subsample removal (Remark~\ref{rem:review}) removes this attack vector \\
X5c coordinate attack ($n{=}16, f{=}3$) & pass & the Mhamdi coordinate-concentrated attack (whole budget on one axis, within the honest Euclidean spread) is \emph{selected} by Krum, all 3 of 3 Byzantine vectors, in each of the two seeded regime instances (colluding / non-colluding): plain multi-Krum shows a coordinate deviation of $2.29$ (colluding) / $2.35$ (non-colluding), a \emph{bounded} $O(\sqrt{d})$ Krum-level bias within the imported envelope, not a violation, while Bulyan's per-coordinate median-trim drives it to $0.00$ (in both colluding and non-colluding regimes), because the $\theta{=}n{-}2f$ selected are honest-majority per coordinate \\
X8 guarantee-class contrast & pass & on non-comparable objects (multi-round gossip parameters vs.\ a single replicated aggregate), gossip-averaging SGD leaves replicas $\approx 0.09$ apart in $L_\infty$ at quiescence while \ACFA{} agrees to $0$ bytes, a guarantee-class contrast (statistical vs.\ bitwise), not a performance comparison \\
X6 quantisation cost & pass & on one learnable task, fixed-point and float robust-aggregation training reach the same held-out accuracy ($0.985$, three significant figures, single seed); this bounds the $Q16.16$ quantisation cost below the metric resolution for that task only, not in general \\
A4 dedup ablation & diverges & removing the per-round uniqueness clause lets an equivocator split admitted sets between proof-lagged replicas $\rightarrow$ divergent roots \\
\bottomrule
\end{tabular}
\end{table}

Two points are worth stating plainly. The X6 delta is genuinely zero only because
$Q16.16$ is lossless \emph{for this task}; it is not a general claim that
quantisation is free, and X5 shows the kernel tracks its float reference rather
than being identical to it. And X9 is the test the adversarial review demanded:
the earlier subsampled kernel was grindable, and only the full-dimension rule
(Remark~\ref{rem:review}) withstands a last-mover search, but robustness
remains an \emph{imported} guarantee under the quantisation-margin condition of
Theorem~\ref{thm:composed}, Part~B, not a contribution of this paper.

\paragraph{The quantisation-margin measurement.} To test
Lemma~\ref{lem:margin} we swept $1600$ configurations (honest random and
near-duplicate-value, each at two weight scales; fixed RNG seed $3$, and the
qualitative three-tier resolution (Remark~\ref{rem:margin}) is the claim, the exact band is
seed-conditional), comparing the $Q16.16$ multi-Krum selection to
a finer-precision reference and recording the boundary margin $g$ against the
worst-case perturbation $2\beta$. There were \textbf{zero selection flips} in all
$1600$; the sufficient condition $g > 2\beta$ held in $95.5$--$100\%$ of
configurations with a median safety factor of $11.5$ to $346$ (a tighter
expected-error $\beta$ raises this to $99.25$--$100\%$), and even the residual
near-tie configurations, where the worst-case $\beta$ leaves the condition
inconclusive, flipped nothing. We do not claim the uncovered tail is harmless;
by Lemma~\ref{lem:discont} a flip there would move the aggregate by $\Delta/m$,
we report only that none occurred across the $1600$ configurations, including
the tight-scale bucket whose minimum observed boundary gap is
$\approx 1.0\times 10^{-5}$, the genuine boundary probe (Remark~\ref{rem:margin}).

\subsection{Observations on the Determinism Envelope}\label{sec:ablations}

Two incidental measurements from constructing ablation A3 bear on how
floating-point non-determinism manifests in practice. CPython~3.12+
implements the built-in \texttt{sum} with compensated (Neumaier)
summation~\cite{neumaier1974rundungsfehler,cpython2022sum}, so
a naive-versus-\texttt{fsum} comparison shows no divergence; the divergence
appears in plain loop accumulation, the order-dependent pattern of C and BLAS
kernels~\cite{goldberg1991floating,higham1993accuracy},
where 39 of the 45 pairwise distance sums differ between two accumulation
orders in the artifact instance (the count is emitted by the battery itself). Separately, values quantised to $Q16.16$ and then
converted to doubles sum \emph{exactly} in floating point regardless of order,
because all terms lie on a common dyadic grid with every partial sum's numerator below $2^{53}$ (comfortably satisfied at this scale);
quantisation alone therefore buys order-invariance for sums, though not for
the full nonlinear kernel. Neither observation substitutes for the integer
kernel; both delimit where the hazard lives. $Q16.16$ also fixes the dynamic
range at $\pm 2^{15}$ with $2^{-16}$ resolution; gradient components far below
that resolution quantise to zero, so the format is a deployment parameter, not a
universal choice. Any per-tensor scaling introduced to widen the range must
itself be a deterministic state-derived function; otherwise it reopens the
exact non-determinism the integer kernel exists to close.

\section{Discussion}\label{sec:discussion}

\paragraph{Relation to the two-layer architecture.} \ACFA{} is the L4
extension sketched in~\cite{gillespie2026crdtmerge} made concrete. Where that
architecture assumed environmental determinism (its Assumption~10), the present
construction discharges the obligation with integer arithmetic, which
Lemma~\ref{lem:discont} shows is not optional for selection rules.

\paragraph{Relation to E4, and the reduction-class distinction.} A closer
relative is the author's own E4 protocol
(\texttt{crdt-merge}~v0.9.5+;~\cite{gillespie2026e4}), which realises the
single-shared-lattice form of that sketch: a multi-dimensional trust lattice in
which trust propagates as CRDT data through the same pipeline it validates. E4
already provides several primitives this paper also uses, namely cryptographically
self-authenticating equivocation evidence (two conflicting signed operations from
one peer at the same sequence number), a canonical-order deferral of an
order-dependent resolver, and a fixed-point integer accumulation mode, so E4
is prior work this paper builds on, not around, and the distinction is precise.
It lies in the \emph{reduction class} and the \emph{actuation}, not in the use of
proofs. E4's resolver is a \emph{trust-weighted mean with outlier detection},
a smooth, order-independent aggregate. For a smooth aggregate, deterministic
quantisation of the parameters, and full-precision re-anchoring when accumulated
rounding drifts past a threshold, suffice for convergence by the \emph{tolerance}
route. \ACFA{}'s reduction is instead a distance-geometry \emph{selector}
(multi-Krum), which Lemma~\ref{lem:discont} shows is \emph{discontinuous}. A
bounded perturbation can flip the selection and move the aggregate by
$\Delta/m$. For such a selector the tolerance route is closed: no
quantise-and-re-anchor scheme can guarantee agreement, because the disagreement
is not a bounded numeric drift but a discrete selection flip. This is why
\ACFA{} discharges determinism with \emph{exact integer arithmetic} rather than
quantisation-with-re-anchoring, a different mechanism forced by the selection
class, not a refinement of the same one. Likewise, where E4 converts equivocation
evidence into a \emph{graded, progressive} trust reduction that feeds the
trust-weighted mean, \ACFA{} converts the same evidentiary kind of proof into a
\emph{binary, permanent} conviction that gates admission, a different effect
suited to a selector that has no notion of a soft weight. The evidentiary
primitive is shared; the aggregation it feeds, and the arithmetic that
aggregation forces, are what differ.

\paragraph{Multi-round execution and epoch forking.} The guarantee is
per-round. Deciding that round $r$ is \emph{closed} is quiescence detection
(a liveness question a CRDT alone does not answer), and replicas that cut at
different points fork the round-DAG until they reconcile
(Section~\ref{sec:intro}): one replica may train on round-$r$ aggregate $Y_r$
while another, having admitted a late proof, has re-resolved to $Y_r'$ and moved
on. Per-round safety is unaffected by the cut choice (any two replicas that
close $r$ on the same contribution set compute the same bytes), but multi-round
\emph{liveness} requires an orthogonal epoch-synchronisation mechanism
(synchronous timeouts, threshold signatures on round closure, or any
transport-layer commitment) which this paper explicitly delegates. Such a
mechanism agrees only on \emph{when} a round closes, never on the aggregate
\emph{value}, which stays a consensus-free pure function of the admitted set; the
coordination removed is agreement on the computation, and only round-closure
timing is delegated.

\paragraph{The adaptive margin adversary.} Lemma~\ref{lem:margin} is a
worst-case condition against \emph{static} inputs; an adaptive last-mover that
observes honest contributions could compute the exact integer scores and craft a
vector forcing $g \approx 0$, landing the configuration in the uncovered tier
of Remark~\ref{rem:margin}. Two structural features mitigate this, though
neither closes it. First, for a \emph{single-shot} adversary, ties are broken by
content hash, a state-canonical order it cannot steer independently of its own
vector's bytes, so forcing a tie costs it control over \emph{which} side wins,
degrading a targeted flip to an untargeted $\Delta/m$ perturbation, which is what
the imported envelope already prices in. A \emph{grinding} last-mover, however,
may search jointly over the tie constraint \emph{and} its own resulting hash, and
so is not excluded by this argument; that residual is exactly the item
Remark~\ref{rem:margin} names as open (a closed-form anti-concentration bound on
the score-gap distribution). Second, the distinct get-\emph{admitted} grinding
variant, a last-mover searching candidates to place itself in the selection,
is measured and closed (X9: $0/900$ selections).

\paragraph{Provable versus statistical exclusion.} The construction separates
two adversary classes cleanly. Equivocation is \emph{convicted} provably and
evicted \emph{permanently}, with no vote; malformed signatures are excluded by
verification without conviction; per-round duplicates are excluded by the
uniqueness clause, likewise without conviction. Adversaries that stay within protocol and submit
poisoned values are bounded \emph{statistically} by the robust rule, per round.
The composition is strictly stronger than per-round filtering alone (an
equivocator is removed from all future rounds, not re-filtered each round),
but it is not a claim that all Byzantine behaviour is provable: a
within-bounds poisoner leaves no proof.

\paragraph{Costs.} Relative to a coordinator, the costs are full-state or
delta gossip of the contribution set, one signature per contribution and two
verifications per proof; the fixed-point kernel (integer multi-Krum is
$O(n^2 d)$~\cite{blanchard2017machine} like its float counterpart); and quantisation of contributions to
$Q16.16$. Relative to consensus-based decentralisation, there are no rounds,
no quorums, and no leader; the price is the weaker guarantee, eventual,
not immediate, agreement.

\subsection{Limitations and Future Work}\label{sec:limitations}

The robustness leg (Theorem~\ref{thm:composed}, Part~B) is the weakest part and the
adversarial review sharpened it. (a) \emph{The quantisation-margin condition is
the real open problem}: because multi-Krum is discontinuous
(Lemma~\ref{lem:discont}), $Q16.16$ rounding can flip the selection near a score tie, moving the
aggregate by $\Delta/m$, so the imported over-the-reals bound holds only
when the honest score margin dominates the quantisation perturbation; bounding
the flip probability under a margin distribution is open, and is not the additive
$\varepsilon$ an earlier draft implied. (b) \emph{Robustness needs an admitted-population condition} ($|A| \ge 2f+3$ among admitted contributions), a
delivery/completeness condition that the consistency leg (Part~A) does not require;
we separate the two so the consensus-free claim is not overstated. (c) \emph{The default rule is Krum-level, not Bulyan-level}: full-dimension
multi-Krum is vulnerable to coordinate-concentrated attacks~\cite{mhamdi2018hidden}
(Section~\ref{sec:extbattery}, X5c); Bulyan restores coordinate robustness but
needs $|A| \ge 4f+3$, a stronger admitted-population condition than the $|A| \ge 2f+3$ of the default. The
consistency guarantee is unaffected by the choice. More fundamentally, Krum and
Bulyan are themselves defeated by within-norm attacks that stay inside the honest
spread~\cite{baruch2019little,shejwalkar2021manipulating}, exactly the stealth
drift X5b reports as unfiltered. This is a property of the imported rule, not the
framework; \ACFA{} replicates whatever robust rule is plugged in, so a stronger
successor to Krum inherits the same consistency and accountability guarantees
unchanged. (d)
\emph{Cross-architecture byte-identity is argued, not measured}: the integer
kernel is exact under a fixed width contract, but the prototype runs one
architecture; a heterogeneous run (or a fixed-width port with documented overflow
semantics) is owed. The prototype validates correctness, not performance: no
timing, bandwidth, or large-$n$ measurements are reported, and delta-state
propagation~\cite{almeida2018delta} would be required at scale, as in the
parent architecture. The equivocation G-Set grows monotonically; rate-limiting
proof spam (all of it invalid by construction, but bandwidth-consuming) and
garbage-collecting evidence for departed identities are engineering questions
we have not addressed. Sybil resistance is delegated to the PKI. Finally, the
accountable class here is equivocation; extending self-authenticating evidence
to other violation classes (for example, provably malformed training
provenance) is future work.

\section{Related Work}\label{sec:related}

\paragraph{Byzantine-robust aggregation, central and decentralised.}
Krum~\cite{blanchard2017machine}, trimmed mean~\cite{yin2018byzantine},
Bulyan~\cite{mhamdi2018hidden}, and FLTrust~\cite{cao2021fltrust} define and
analyse robust rules at a coordinator. This paper does not modify the rules or
their bounds; it replicates their evaluation without the coordinator. The
difficulty of doing so is known: applying geometry-based selectors in a
decentralised topology makes honest nodes filter \emph{different} peer sets and
therefore disagree. Wu, Chen and Ling~\cite{wu2023byzantine} name this
``disagreement'' precisely and address it with an iterative outlier scissor (IOS)
scheme carrying optimisation-convergence guarantees;
Cajaraville-Aboy et al.~\cite{cajaraville2025byzantine} (WFAgg) combine complementary
distance-, cosine-, and temporal-based filters that individually miss different
attacks, and propose a weighting aggregator;
and BALANCE~\cite{fang2024byzantine} gives a decentralised Byzantine-robust
method with convergence guarantees. Closest in surface form is
BRACE~\cite{fang2025brace}, which distributes one quantised aggregate through the
ring's lockstep reduce with a coordinate-wise \emph{consensus} on the sign and a
convergence guarantee; any agreement there is a product of the consensus this
work removes, not of a pure function of converged state. We do not claim the
problem is open. Our
distinction is the \emph{guarantee class}: those works achieve statistical or
optimisation convergence in a neighbour-gossip topology, where honest replicas end
near one another. \ACFA{} instead achieves \emph{exact byte-identical} agreement
on the robust aggregate over a replicated contribution set, via deferred
deterministic replay (the accountability engine of
PeerReview~\cite{haeberlen2007peerreview}, here discharged for a discontinuous
selector by exact integer arithmetic) plus provable accountability. Consistency, not convergence
rate, is the object.

\paragraph{Byzantine fault tolerance in CRDTs.}
Kleppmann~\cite{kleppmann2022byzantine} makes CRDT \emph{state} tolerate
Byzantine replicas via hash-graph techniques; the object protected is the
replicated set itself. The blocklace~\cite{almeida2024blocklace} goes further,
using self-authenticating signed-hash evidence to detect equivocation and
eventually \emph{exclude} equivocators inside a CRDT without consensus, the
exclusion-by-proof primitive our accountability leg reuses; concurrently,
Brocco~\cite{brocco2026composable} obtains Byzantine-resilient convergence by
deferring a deterministic interpretation over accumulated CRDT state, and a
concurrent follow-up~\cite{brocco2026decoupling} decouples identity-based from
content-based trust for fine-grained post-compromise handling; our admission
gate is deliberately coarser, a \emph{binary, permanent} identity conviction,
not a graded per-update trust, because the discontinuous selector it gates
admits no soft weight. The
present work is complementary to all three: it assumes eventual pairwise
delivery among honest replicas for state propagation (the $|A| \ge 2f+3$ bound
constrains the contribution population, not propagation) and protects a
non-associative \emph{reduction over} the state, gating that reduction's
admission by the exclusion proof, an object those constructions do not
protect: the blocklace and Brocco's layer defend convergence of the replicated
state or of a benign interpretation of it, whereas the selector here is
discontinuous in and adversarially coupled across its inputs. Where a subclass of CRDTs can
simply \emph{tolerate} equivocation and retain strong eventual
consistency~\cite{jacob2021equivocation}, our reduction cannot: two admitted
entries from one equivocator would flip the discontinuous selector, so
equivocation must be \emph{excluded}, not absorbed. Accountable
consensus~\cite{civit2021polygraph} detects and proves misbehaviour
\emph{within} agreement protocols; we import the self-authenticating-evidence
idea into a setting with no agreement protocol at all.

\paragraph{Decentralised training.}
Gossip SGD~\cite{lian2017decentralized,koloskova2019decentralized} and
low-communication schemes~\cite{douillard2023diloco} target statistical
convergence of training; replicas hold similar, not identical, parameters, and
Byzantine handling (where present) is per-replica and unsynchronised.
Federated learning~\cite{mcmahan2017communication,kairouz2021advances}
centralises aggregation. The guarantee class here, bitwise agreement on the
robust aggregate, coordinator-free, is different from both, and is the
guarantee the state-convergence line~\cite{shapiro2011conflict,
gillespie2026crdtmerge} provides once the reduction is made a pure function of
converged state.

\paragraph{Composition of CRDTs with queries.}
That a deterministic query of a converged CRDT is replica-consistent is
implicit in the OR-Set definition~\cite{shapiro2011comprehensive}, which builds
query-agreement into the datatype rather than stating it as a separate result;
our prior work formalised it for one
OR-Set and a deterministic strategy~\cite{gillespie2026crdtmerge}. Which
computations over a lattice are order-independent has a longer lineage. LVars
restricts reads to monotone threshold queries~\cite{kuper2013lvars}, and the
CALM line ties coordination-freedom to
monotonicity~\cite{conway2012logic,hellerstein2020calm}; our converse
(Theorem~\ref{thm:iff}) instead admits \emph{non-monotone} reductions by
deferring to a converged (quiesced) set, characterising consistency by
permutation-invariance plus state-derived entropy rather than by monotonicity.
The present
contribution is the product form with an evidence lattice, the explicit
admission of non-monotone stochastic reductions with state-derived entropy,
the converse (Theorem~\ref{thm:iff}), and the demonstration that this suffices
for an established class of Byzantine-robust rules.

\section{Conclusion}\label{sec:conclusion}

Byzantine-robust aggregation rules are globally coupled, non-associative,
and discontinuous, three properties that appear to demand a
coordinator or a consensus layer. They do not. The rules need an agreed set
and an agreed exclusion predicate; both are join-semilattices, both converge
without coordination, and everything downstream is a pure function of their
product, which Theorem~\ref{thm:lifting} makes replica-consistent wholesale.
The contribution is a composition and an application, not a new algebra: the
product-lifting theorem is elementary (Theorem~\ref{thm:lifting}) and its value
is the converse (Theorem~\ref{thm:iff}) plus the demonstration that a
discontinuous, globally-coupled robust selector can be made
strong-eventually-consistent with exact arithmetic and accountable eviction.
Consensus-free accountable eviction is itself established: self-healing lattice
agreement evicts faulty replicas over a lattice via undeniable
proofs~\cite{freitas2021selfhealing}, and the blocklace excludes equivocators
inside a CRDT by self-authenticating evidence~\cite{almeida2024blocklace}. What
we did not find in the Byzantine-CRDT or Byzantine-FL prior art is that primitive
composed with a discontinuous, globally-coupled robust \emph{selector} as the
gated object, under the exact-arithmetic determinism obligation that
discontinuity forces.
The two obligations that remain, exact arithmetic, because discontinuity
voids any tolerance (Lemma~\ref{lem:discont}), and state-derived entropy,
because the converse theorem permits nothing else
(Theorem~\ref{thm:iff}), are both dischargeable by construction, and the
16-check falsification battery, including three ablations that break
byte-identity by removing one mechanism each (a fourth, A4, in the extended
battery), locates the load-bearing points
empirically. Equivocation is handled by evidence rather than by vote: a proof
any party can check offline, accumulated in a grow-only lattice, evicting
deterministically even when it arrives after the fact. The result is a
coordinator-free replication discipline for exactly the aggregation rules that
adversarial distributed learning already uses, with the scope limits stated
in Section~\ref{sec:intro}: consistency not accuracy, safety not liveness,
and proofs only for what is provable.

\paragraph{Reproducibility.} The prototype and all test harnesses will be
released publicly upon publication and are available from the author on request
in the interim. Every reported number is regenerated by a fully seeded script: \texttt{python3 falsifier.py}
reproduces the 16-check transcript of Table~\ref{tab:falsifier} (and Appendix~\ref{app:falsifier});
\texttt{python3 battery.py} the extended battery of Table~\ref{tab:battery}
(rows X1, X5a, X5b, X6, X8, A4), with \texttt{python3 x9\_sweep.py} regenerating
the X9 grinding sweep ($0/900$; $300$ trials) and \texttt{python3 x5c\_bulyan.py}
the X5c coordinate-attack row ($2.29/2.35$, Bulyan $0.00$); and
\texttt{python3 t5\_flipprob.py} the quantisation-margin sweep of
Section~\ref{sec:extbattery} (RNG seed $3$).

\bibliographystyle{plain}
\bibliography{references}

@article{almeida2018delta,
  author    = {Paulo S\'ergio Almeida and Ali Shoker and Carlos Baquero},
  title     = {Delta state replicated data types},
  journal   = {Journal of Parallel and Distributed Computing},
  volume    = {111},
  pages     = {162--173},
  year      = {2018},
  doi       = {10.1016/j.jpdc.2017.08.003},
}

@inproceedings{blanchard2017machine,
  author    = {Peva Blanchard and El Mahdi El Mhamdi and Rachid Guerraoui and Julien Stainer},
  title     = {Machine Learning with Adversaries: Byzantine Tolerant Gradient Descent},
  booktitle = {Advances in Neural Information Processing Systems 30 (NeurIPS)},
  pages     = {119--129},
  year      = {2017},
}

@article{kairouz2021advances,
  author    = {Peter Kairouz and H. Brendan McMahan and Brendan Avent and Aur\'elien Bellet and Mehdi Bennis and Arjun Nitin Bhagoji and Kallista Bonawitz and Zachary Charles and Graham Cormode and Rachel Cummings and Rafael G. L. D'Oliveira and Hubert Eichner and Salim El Rouayheb and David Evans and Josh Gardner and Zachary Garrett and Adri\`a Gasc\'on and Badih Ghazi and Phillip B. Gibbons and Marco Gruteser and Zaid Harchaoui and Chaoyang He and Lie He and Zhouyuan Huo and Ben Hutchinson and Justin Hsu and Martin Jaggi and Tara Javidi and Gauri Joshi and Mikhail Khodak and Jakub Kone\v{c}n\'y and Aleksandra Korolova and Farinaz Koushanfar and Sanmi Koyejo and Tancr\`ede Lepoint and Yang Liu and Prateek Mittal and Mehryar Mohri and Richard Nock and Ayfer \"Ozg\"ur and Rasmus Pagh and Mariana Raykova and Hang Qi and Daniel Ramage and Ramesh Raskar and Dawn Song and Weikang Song and Sebastian U. Stich and Ziteng Sun and Ananda Theertha Suresh and Florian Tram\`er and Praneeth Vepakomma and Jianyu Wang and Li Xiong and Zheng Xu and Qiang Yang and Felix X. Yu and Han Yu and Sen Zhao},
  title     = {Advances and Open Problems in Federated Learning},
  journal   = {Foundations and Trends in Machine Learning},
  volume    = {14},
  number    = {1--2},
  pages     = {1--210},
  year      = {2021},
  doi       = {10.1561/2200000083},
}

@inproceedings{koloskova2019decentralized,
  author    = {Anastasia Koloskova and Sebastian U. Stich and Martin Jaggi},
  title     = {Decentralized Stochastic Optimization and Gossip Algorithms with Compressed Communication},
  booktitle = {Proceedings of the 36th International Conference on Machine Learning (ICML)},
  volume    = {97},
  pages     = {3478--3487},
  year      = {2019},
}

@inproceedings{lian2017decentralized,
  author    = {Xiangru Lian and Ce Zhang and Huan Zhang and Cho-Jui Hsieh and Wei Zhang and Ji Liu},
  title     = {Can Decentralized Algorithms Outperform Centralized Algorithms? {A} Case Study for Decentralized Parallel Stochastic Gradient Descent},
  booktitle = {Advances in Neural Information Processing Systems 30 (NeurIPS)},
  year      = {2017},
}

@inproceedings{mcmahan2017communication,
  author    = {H. Brendan McMahan and Eider Moore and Daniel Ramage and Seth Hampson and Blaise Ag\"uera y Arcas},
  title     = {Communication-Efficient Learning of Deep Networks from Decentralized Data},
  booktitle = {Proceedings of the 20th International Conference on Artificial Intelligence and Statistics (AISTATS)},
  pages     = {1273--1282},
  year      = {2017},
}

@article{sanjuan2020merkle,
  author    = {Hector Sanjuan and Samuli Poyhtari and Pedro Teixeira and Ioannis Psaras},
  title     = {{Merkle-CRDTs}: {Merkle-DAGs} Meet {CRDTs}},
  journal   = {arXiv preprint arXiv:2004.00107},
  year      = {2020},
}

@techreport{shapiro2011comprehensive,
  author      = {Marc Shapiro and Nuno Pregui\c{c}a and Carlos Baquero and Marek Zawirski},
  title       = {A Comprehensive Study of Convergent and Commutative Replicated Data Types},
  institution = {INRIA},
  number      = {RR-7506},
  year        = {2011},
}

@inproceedings{shapiro2011conflict,
  author    = {Marc Shapiro and Nuno Pregui\c{c}a and Carlos Baquero and Marek Zawirski},
  title     = {Conflict-Free Replicated Data Types},
  booktitle = {Proceedings of the 13th International Symposium on Stabilization, Safety, and Security of Distributed Systems (SSS)},
  series    = {Lecture Notes in Computer Science},
  volume    = {6976},
  pages     = {386--400},
  publisher = {Springer},
  year      = {2011},
  doi       = {10.1007/978-3-642-24550-3_29},
}

@article{vogels2009eventually,
  author    = {Werner Vogels},
  title     = {Eventually Consistent},
  journal   = {Communications of the ACM},
  volume    = {52},
  number    = {1},
  pages     = {40--44},
  year      = {2009},
  doi       = {10.1145/1435417.1435432},
}

@inproceedings{yin2018byzantine,
  author    = {Dong Yin and Yudong Chen and Kannan Ramchandran and Peter L. Bartlett},
  title     = {Byzantine-Robust Distributed Learning: Towards Optimal Statistical Rates},
  booktitle = {Proceedings of the 35th International Conference on Machine Learning (ICML)},
  pages     = {5650--5659},
  year      = {2018}
}

@inproceedings{mhamdi2018hidden,
  author    = {El Mahdi El Mhamdi and Rachid Guerraoui and S{\'e}bastien Rouault},
  title     = {The Hidden Vulnerability of Distributed Learning in {Byzantium}},
  booktitle = {Proceedings of the 35th International Conference on Machine Learning (ICML)},
  pages     = {3521--3530},
  year      = {2018}
}

@inproceedings{kleppmann2022byzantine,
  author    = {Martin Kleppmann},
  title     = {Making {CRDTs} {Byzantine} Fault Tolerant},
  booktitle = {Proceedings of the 9th Workshop on Principles and Practice of Consistency for Distributed Data (PaPoC)},
  pages     = {8--15},
  year      = {2022}
}

@inproceedings{cao2021fltrust,
  author    = {Xiaoyu Cao and Minghong Fang and Jia Liu and Neil Zhenqiang Gong},
  title     = {{FLTrust}: {Byzantine}-Robust Federated Learning via Trust Bootstrapping},
  booktitle = {Proceedings of the Network and Distributed System Security Symposium (NDSS)},
  year      = {2021}
}

@article{dolev1986approximate,
  author  = {Danny Dolev and Nancy A. Lynch and Shlomit S. Pinter and Eugene W. Stark and William E. Weihl},
  title   = {Reaching Approximate Agreement in the Presence of Faults},
  journal = {Journal of the ACM},
  volume  = {33},
  number  = {3},
  pages   = {499--516},
  year    = {1986}
}

@misc{douillard2023diloco,
  author = {Arthur Douillard and Qixuan Feng and Andrei A. Rusu and Rachita Chhaparia and Yani Donchev and Adhiguna Kuncoro and Marc'Aurelio Ranzato and Arthur Szlam and Jiajun Shen},
  title  = {{DiLoCo}: Distributed Low-Communication Training of Language Models},
  note   = {arXiv preprint arXiv:2311.08105},
  year   = {2023}
}

@inproceedings{civit2021polygraph,
  author    = {Pierre Civit and Seth Gilbert and Vincent Gramoli},
  title     = {Polygraph: Accountable {Byzantine} Agreement},
  booktitle = {Proceedings of the 41st IEEE International Conference on Distributed Computing Systems (ICDCS)},
  pages     = {403--413},
  doi       = {10.1109/ICDCS51616.2021.00046},
  year      = {2021}
}

@article{wu2023byzantine,
  author  = {Zhaoxian Wu and Tianyi Chen and Qing Ling},
  title   = {Byzantine-Resilient Decentralized Stochastic Optimization With Robust Aggregation Rules},
  journal = {IEEE Transactions on Signal Processing},
  volume  = {71},
  pages   = {3179--3195},
  year    = {2023},
  doi     = {10.1109/TSP.2023.3300629}
}

@article{cajaraville2025byzantine,
  author  = {Diego Cajaraville-Aboy and Ana Fern{\'a}ndez-Vilas and Rebeca P. D{\'i}az-Redondo and Manuel Fern{\'a}ndez-Veiga},
  title   = {Byzantine-Robust Aggregation for Securing Decentralized Federated Learning},
  journal = {IEEE Access},
  volume  = {13},
  pages   = {190947--190963},
  year    = {2025},
  doi     = {10.1109/ACCESS.2025.3629864}
}

@inproceedings{fang2024byzantine,
  author    = {Minghong Fang and Zifan Zhang and Hairi and Prashant Khanduri and Jia Liu and Songtao Lu and Yuchen Liu and Neil Zhenqiang Gong},
  title     = {Byzantine-Robust Decentralized Federated Learning},
  booktitle = {Proceedings of the 2024 ACM SIGSAC Conference on Computer and Communications Security (CCS)},
  year      = {2024},
  doi       = {10.1145/3658644.3670307}
}

@misc{gillespie2026crdtmerge,
  author = {Ryan Gillespie},
  title  = {Conflict-Free Replicated Data Types for Neural Network Model Merging: A Two-Layer Architecture Enabling {CRDT}-Compliant Model Merging Across 26 Strategies},
  note   = {arXiv preprint arXiv:2605.19373},
  year   = {2026}
}

@misc{gillespie2026e4,
  author = {Ryan Gillespie},
  title  = {Recursive Trust-as-Data Binding for Byzantine Fault Tolerant Conflict-Free Replicated Data Type Synchronisation},
  note   = {UK Patent Application GB2608127.3; implemented in the {\tt crdt-merge} library v0.9.5+},
  year   = {2026}
}

@article{baquero2017composition,
  author  = {Carlos Baquero and Paulo S\'ergio Almeida and Alcino Cunha and Carla Ferreira},
  title   = {Composition in State-based Replicated Data Types},
  journal = {Bulletin of the European Association for Theoretical Computer Science (EATCS)},
  volume  = {123},
  year    = {2017},
}

@inproceedings{kuper2013lvars,
  author    = {Lindsey Kuper and Ryan R. Newton},
  title     = {{LVars}: Lattice-based Data Structures for Deterministic Parallelism},
  booktitle = {Proceedings of the 2nd ACM SIGPLAN Workshop on Functional High-Performance Computing (FHPC)},
  pages     = {71--84},
  year      = {2013},
  doi       = {10.1145/2502323.2502326},
}

@inproceedings{conway2012logic,
  author    = {Neil Conway and William R. Marczak and Peter Alvaro and Joseph M. Hellerstein and David Maier},
  title     = {Logic and Lattices for Distributed Programming},
  booktitle = {Proceedings of the 3rd ACM Symposium on Cloud Computing (SoCC)},
  year      = {2012},
  doi       = {10.1145/2391229.2391230},
}

@article{hellerstein2020calm,
  author  = {Joseph M. Hellerstein and Peter Alvaro},
  title   = {Keeping {CALM}: When Distributed Consistency Is Easy},
  journal = {Communications of the ACM},
  volume  = {63},
  number  = {9},
  pages   = {72--81},
  year    = {2020},
  doi     = {10.1145/3369736},
}

@article{almeida2024blocklace,
  author  = {Paulo S\'ergio Almeida and Ehud Shapiro},
  title   = {The Blocklace: A {Byzantine}-repelling and Universal Conflict-free Replicated Data Type},
  journal = {arXiv preprint arXiv:2402.08068},
  year    = {2024},
}

@inproceedings{freitas2021selfhealing,
  author    = {Luciano Freitas de Souza and Petr Kuznetsov and Thibault Rieutord and Sara Tucci-Piergiovanni},
  title     = {Accountability and Reconfiguration: Self-Healing Lattice Agreement},
  booktitle = {25th International Conference on Principles of Distributed Systems (OPODIS)},
  series    = {LIPIcs},
  volume    = {217},
  pages     = {25:1--25:23},
  year      = {2021},
  doi       = {10.4230/LIPIcs.OPODIS.2021.25},
}

@inproceedings{haeberlen2007peerreview,
  author    = {Andreas Haeberlen and Petr Kuznetsov and Peter Druschel},
  title     = {{PeerReview}: Practical Accountability for Distributed Systems},
  booktitle = {Proceedings of the 21st ACM SIGOPS Symposium on Operating Systems Principles (SOSP)},
  pages     = {175--188},
  year      = {2007},
  doi       = {10.1145/1294261.1294279},
}

@article{garciamarquez2025kfc,
  author  = {Mario Garc\'ia-M\'arquez and Nuria Rodr\'iguez-Barroso and M. Victoria Luz\'on and Francisco Herrera},
  title   = {Krum Federated Chain ({KFC}): Using Blockchain to Defend Against Adversarial Attacks in Federated Learning},
  journal = {arXiv preprint arXiv:2502.06917},
  year    = {2025},
}

@article{brocco2026composable,
  author  = {Amos Brocco},
  title   = {A Composable {CRDT} Layer for {Byzantine}-Resilient Deterministic Reconstruction},
  journal = {arXiv preprint arXiv:2606.18966},
  year    = {2026},
}

@article{brocco2026decoupling,
  author  = {Amos Brocco},
  title   = {Decoupling Trust in {Byzantine} {CRDT}s: Fine-grained Post-Compromise Handling without Breaking Causality},
  journal = {arXiv preprint arXiv:2606.31759},
  year    = {2026},
}

@article{jacob2021equivocation,
  author  = {Florian Jacob and Saskia Bayreuther and Hannes Hartenstein},
  title   = {On Conflict-Free Replicated Data Types and Equivocation in {Byzantine} Setups},
  journal = {arXiv preprint arXiv:2109.10554},
  year    = {2021},
}

@article{buterin2017casper,
  author  = {Vitalik Buterin and Virgil Griffith},
  title   = {Casper the Friendly Finality Gadget},
  journal = {arXiv preprint arXiv:1710.09437},
  year    = {2017},
}

@article{goldberg1991floating,
  author  = {David Goldberg},
  title   = {What Every Computer Scientist Should Know About Floating-Point Arithmetic},
  journal = {ACM Computing Surveys},
  volume  = {23},
  number  = {1},
  pages   = {5--48},
  year    = {1991},
  doi     = {10.1145/103162.103163},
}

@article{higham1993accuracy,
  author  = {Nicholas J. Higham},
  title   = {The Accuracy of Floating Point Summation},
  journal = {SIAM Journal on Scientific Computing},
  volume  = {14},
  number  = {4},
  pages   = {783--799},
  year    = {1993},
  doi     = {10.1137/0914050},
}

@article{neumaier1974rundungsfehler,
  author  = {Arnold Neumaier},
  title   = {Rundungsfehleranalyse einiger Verfahren zur Summation endlicher Summen},
  journal = {Zeitschrift f\"ur Angewandte Mathematik und Mechanik (ZAMM)},
  volume  = {54},
  number  = {1},
  pages   = {39--51},
  year    = {1974},
  doi     = {10.1002/zamm.19740540106},
}

@misc{cpython2022sum,
  author       = {{Python Software Foundation}},
  title        = {Improve Accuracy of Builtin sum() for Float Inputs (Compensated Summation, {CPython}~3.12)},
  howpublished = {CPython issue 100425, github.com/python/cpython/issues/100425},
  year         = {2022},
}

@inproceedings{baruch2019little,
  author    = {Gilad Baruch and Moran Baruch and Yoav Goldberg},
  title     = {A Little Is Enough: Circumventing Defenses For Distributed Learning},
  booktitle = {Advances in Neural Information Processing Systems 32 (NeurIPS)},
  year      = {2019},
}

@inproceedings{shejwalkar2021manipulating,
  author    = {Virat Shejwalkar and Amir Houmansadr},
  title     = {Manipulating the {Byzantine}: Optimizing Model Poisoning Attacks and Defenses for Federated Learning},
  booktitle = {Proceedings of the Network and Distributed System Security Symposium (NDSS)},
  year      = {2021},
}

@inproceedings{fang2025brace,
  author    = {Minghong Fang and Zhuqing Liu and Xuecen Zhao and Jia Liu},
  title     = {Byzantine-Robust Federated Learning over Ring-All-Reduce Distributed Computing},
  booktitle = {Companion Proceedings of the ACM Web Conference 2025 (WWW Companion)},
  pages     = {961--965},
  year      = {2025},
  doi       = {10.1145/3701716.3715491},
}

\appendix

\section{Falsifier Transcript Summary}\label{app:falsifier}

The battery runs as \texttt{python3 falsifier.py} (CPU only; about four seconds
on commodity hardware). All test-generation seeds are fixed in the script (seed 42 for the battery;
the A3 floating-point instance uses its own fixed seed); the kernel consumes no
entropy beyond the state-derived content-hash tie-break. The Lemma~\ref{lem:nonassoc}
instance is emitted by check X0.3: with $f=1$ over six seeded 200-dimensional
Gaussian vectors, global multi-Krum selects indices $\{1,3,5\}$ while the
fold path (select over the first five, then select over the selection plus the
sixth) yields a different subset \emph{of the same cardinality}, the
fallback regime (second-stage $m=0$) that Lemma~\ref{lem:nonassoc} identifies as
the demanding case for a fold. The A3 ablation instance uses unquantised
Gaussian doubles: with plain loop accumulation, 39 of the 45 pairwise
distance sums differ between arrival-order and sorted-order accumulation, and
the hashed resolve transcripts diverge; the same comparison under CPython's
compensated built-in \texttt{sum}, or over $Q16.16$-quantised doubles (a
common dyadic grid), shows no divergence, as reported in
Section~\ref{sec:ablations}.

\end{document}